\documentclass[10pt]{article} 

\usepackage[utf8]{inputenc} 

\usepackage{graphicx} 
\usepackage{amsthm}

\usepackage{geometry} 
\usepackage{graphicx} 

\usepackage{booktabs} 
\usepackage{array} 
\usepackage{paralist} 
\usepackage{verbatim} 
\usepackage{subfig} 
\usepackage{wrapfig}
\usepackage{amsmath, amsthm, amssymb}
\usepackage{authblk}

\newtheorem{theorem}{Theorem}[section]

\newtheorem{corollary}{Corollary}[section]

\newcommand*{\QED}{\hfill\ensuremath{\square}}

\usepackage{fancyhdr} 
\usepackage{sectsty}
\allsectionsfont{\sffamily\mdseries\upshape} 

\usepackage[nottoc,notlof,notlot]{tocbibind} 
\usepackage[titles,subfigure]{tocloft} 

\title{An Exact, Time-Independent Approach to Clone Size Distributions in Normal and Mutated Cells}
\author[1]{Roshan A}
\author[1]{Jones PH}
\author[2]{Greenman CD\textsection}
\affil[1]{\footnotesize{MRC Cancer Cell Unit, Hutchison-MRC Research Centre, Cambridge, CB2 2XZ, UK}}
\affil[2]{\footnotesize{School of Computing Sciences, University of East Anglia, Norwich, NR4 7TJ, UK}}
\date{} 

\begin{document}
\maketitle

\allowdisplaybreaks

\begin{abstract}

Biological tools such as genetic lineage tracing, 3D confocal microscopy and next generation DNA sequencing are providing new ways to quantify the distribution of clones of normal and mutated cells. Population-wide clone size distributions in vivo are complicated by multiple cell types, and overlapping birth and death processes. This has led to the increased need for mathematically informed models to understand their biological significance. Standard approaches usually require knowledge of clonal age. We show that modelling on clone size independent of time is an alternative method that offers certain analytical advantages; it can help parameterize these models, and obtain distributions for counts of mutated or proliferating cells, for example. When applied to a general birth-death process common in epithelial progenitors this takes the form of a gambler’s ruin problem, the solution of which relates to counting Motzkin lattice paths. Applying this approach to mutational processes, an alternative, exact, formulation of the classic Luria-Delbrück problem emerges. This approach can be extended beyond neutral models of mutant clonal evolution, and also describe some distributions relating to sub-clones within a tumour. The approaches above are generally applicable to any Markovian branching process where the dynamics of different ‘coloured’ daughter branches are of interest.
\end{abstract}

\let\thefootnote\relax\footnote{\textsection Corresponding Author}


\section{Introduction}

One approach to understanding the cellular hierarchy in multicellular organized tissue has been tracking the fate of individual cells either labeled in vivo or isolated ex vivo \cite{Clayton}-\cite{Barrandon}. Improved techniques including genetic lineage tracing and 3D imaging by confocal microscopy have helped further investigate this basic area of research \cite{Kretzschmar}–\cite{Blanpain}. Typically, a cell type of interest is labeled with an identifier and the distribution of its progeny at later time points are observed. Clone distribution data can then be used to decipher division dynamics across the population of cells with great resolution. However, the current methods use population averaging, and are time-dependent posing analytical challenges. There is a need for alternative statistical approaches that may be complementary.

\begin{figure}[t!]
\centering
\includegraphics[width=150mm]{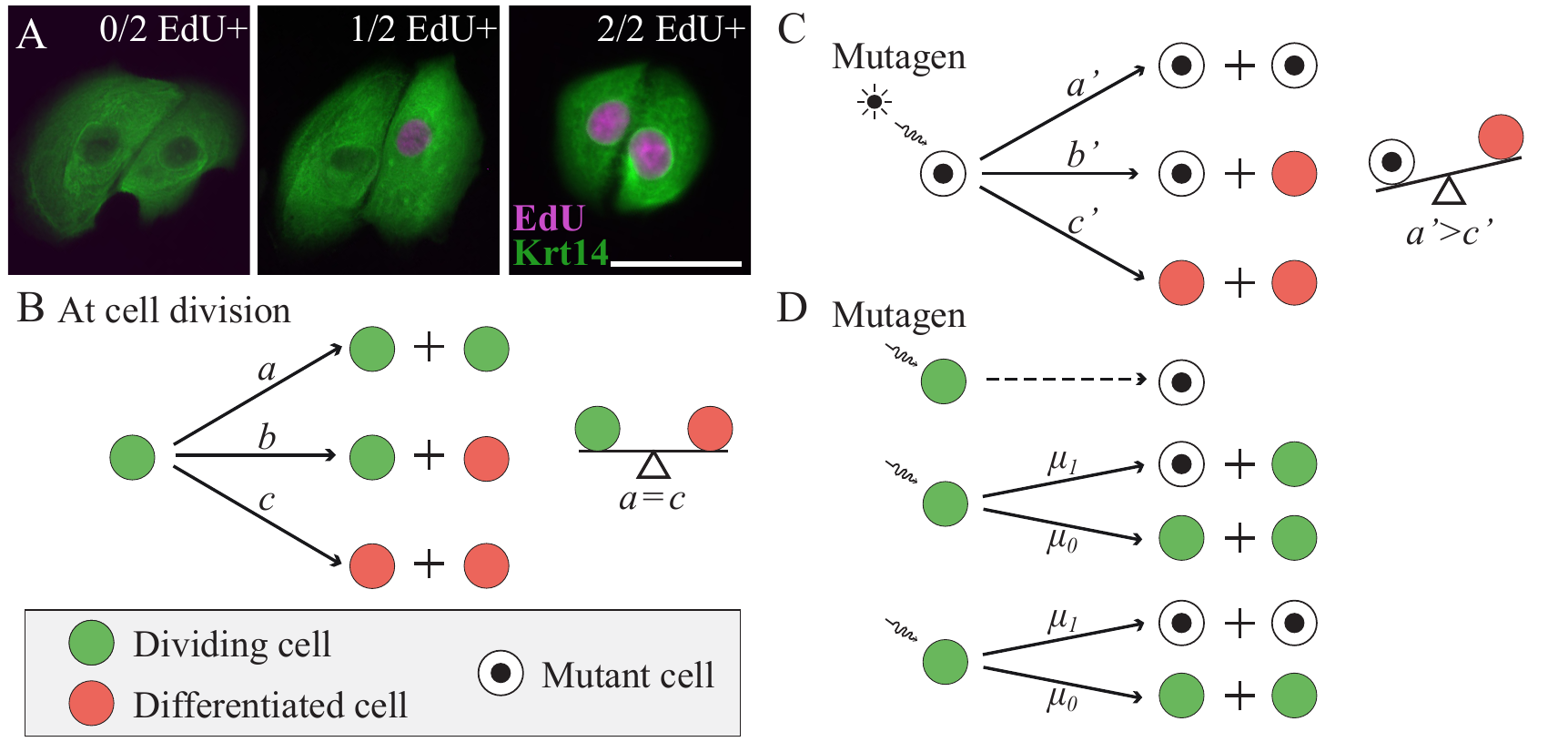}
\caption{Colony formation in normal and mutated cells. (A) Immunofluorescence images of 2-cell clones with the keratinocyte marker Keratin14, and the proliferation marker EdU, showing three possible outcomes of division: two non-proliferating daughters (0/2 EdU+), a non-proliferating and a proliferating daughter (1/2 EdU+), or two proliferating daughters (2/2 EdU+). Scale bar 50$\mu$. (B) Cell division is a birth-death process with three possible outcomes based on the proliferative ability of its daughters. As above, a dividing cell (P) may divide to two dividing daughters (PP), a dividing and differentiated daughter (PD), or two differentiated daughters (DD) in proportions $a$, $b$ and $c$ respectively.  In tissues with high turnover, the number of new dividing cells is equal to the number of non-dividing cells ($a=c$). (C) In the presence of mutagens like UV radiation, this process is imbalanced in p53 mutant clones in favour of proliferation ($a’>c’$). This gives a survival advantage to mutant clones. (D) Mutant cell formation itself is a birth process that can follow one of three possibilities. The first is cell division independent and can occur with background exposure. The second and third possibilities occur following cell division, producing one or two mutant cells out of two daughter cells with probability $\mu_1=1-\mu_0$.}
\label{MutationDivision}
\end{figure}

Adult mammalian epithelium has a high rate of cell division during steady state. Despite this rapid rate of proliferation, the tissue remains in homeostasis as new cells are being generated at the same rate as loss of differentiated cells in a birth-death process ($a=c$ in Figure \ref{MutationDivision}B). A simple illustration of this is in the inter-follicular epidermis, where cell division occurs in the basal layer of a multi-layered epithelium. Cell division here can produce proliferating daughters, that remain in the basal layer, or non-dividing daughters, which are shed to the supra-basal layers, and eventually lost in a process of differentiation. When these keratinocytes are grown in culture, a typical cell division can result in two dividing daughters, one dividing daughter or no dividing daughter out of two total daughters as seen through the uptake of the proliferation marker 5-ethynyl-2'-deoxyuridine (EdU, Figure \ref{MutationDivision}A). Genetic lineage tracing in basal keratinocytes has allowed conditional expression of fluorescent proteins, with all subsequent daughter cells retaining the label, and thus being highlighted as a clone. Clone size distributions thus observed shows maintenance through a population asymmetry of fate outcome in dividing progenitors, with reserve stem cells contributing to wound healing \cite{Clayton},\cite{Mascre},\cite{Klein2}. Additionally, this balance is disturbed in chronic UV irradiation, where p53 mutant keratinocyte clones gain a survival advantage over non-mutant clones mediated through increased proportions of proliferative daughters \cite{Jones} ($a>c$ in Figure \ref{MutationDivision}C). The recent technical advance of live-imaging in epithelia may provide additional information to these models, such as the distribution of cell cycle times \cite{Greco}.

There is also an increasing body of work investigating the growth dynamics of pre-neoplastic and neoplastic tissue \cite{Youssef}–\cite{Barker}. A growing colony of cells can be modeled as a branching process. Luria and Delbrück were the first to produce an analytical examination of the distribution of the number of mutant cells in growing bacterial colonies \cite{LuriaDelbruck}. They used this to show that mutations arise randomly rather than in response to the environment. Their argument was partly deterministic and Lea and Coulson \cite{LeaCoulson} and Bartlett \cite{Bartlett}, \cite{Bartlett2} derived approaches with greater stochastic rigour. These methods generally consider the problem of how many mutants are present after a fixed amount of time. An unpublished combinatorial method by Haldane also exists \cite{Haldane} where all cells divide simultaneously.

These distributions generally assign genes the binary status of mutated or non-mutated. They do not consider the number of mutations in a gene, or the number of different combinations of mutations a subclone of cells may contain. Modern sequencing techniques mean greater resolution of mutations is now possible and there is increased interest in considering distributions associated with combinations of mutations \cite{Nik-Zainal}.

As Kendall observed \cite{Kendall}, \cite{Zheng} there are broadly three models for mutation formulation (Figure \ref{MutationDivision}D). The first formulation would indicate a single cell converts to mutated status at any time independent of the cell division process. This may be the case for continuous exposure to mutagens, such as UV light \cite{Stahl}. The second formulation is the most common formulation where mutations occur in one of the two daughter cells during the cell division process. This is likely to be the case for many mutational processes, where nucleotide errors occur on one of the two DNA strands \cite{Bont}. DNA repair machinery then erroneously corrects this during checkpoints in the cell cycle, resulting in one mutant daughter cell. The third formulation assumes that both daughter cells are mutant. This is also a valid model, and is likely to arise when double stranded breaks occur. When double stranded repair incorrectly repairs the damage, rearrangements result and both daughter cells will be mutant. Some processes such as breakage-fusion-bridge cycles will even result in two mutant daughter cells with distinct rearrangements \cite{Greenman}, \cite{McClintock}. For analytical purposes in this paper, we will assume the most common second formulation. Additionally, we assume that a mutation does not increase the chance of cell loss through apoptosis.

In this work, we consider a different statistical approach to clonal distributions. A standard technique to analyzing a branching process involving two classes of objects, such as mutant/non-mutant, or progenitor/differentiated, is to write down a Chapman-Kolomogorov equation for $P_{m,n}(t)$; the probability of having $m$ and $n$ cells of the two types, at time $t$, and obtain a solution \cite{Kampen}. Instead, we determine the distribution of the number of different types of cells that are present when a fixed number of cells have accumulated, rather than the time that has passed. With this approach, we will see that treating cell differentiation or mutation as time-independent results in exact analytic forms for the distributions of interest. In the next section we obtain the distribution for the number of dividing cells in an epithelial population. We then obtain distributions for the number of mutant cells in a clone undergoing a pure birth process. 


\section{Distribution of Colony Sizes in Homeostatic Tissue}

Tissue homeostasis is balanced by two types of cells; progenitor (dividing) cells ($P$), and differentiated (non-dividing) cells ($D$). As progenitor cells (P) divide, they produce two daughter cells which may be either a progenitor cell or a differentiated cell (D) resulting in the combinations (PP), (PD) or (DD). We assume the probabilities of these occurring are $a$, $b$, and $c$, respectively represented in Figure \ref{MutationDivision}B. Across a population, these probabilities are assumed to be constant, holding the same values for any cell division that takes place at steady state. There is the possibility that apoptosis may form an additional component of this process. Whilst one could incorporate this as an additional branch in the process of Figure \ref{MutationDivision}A, it is assumed negligible in the following analysis.

For simplicity, we assume that we start with a single dividing cell. We also assume the number of descendant cells can be observed, but that (P) and (D) cells can not be distinguished. There are two problems we would like to consider. Firstly, if we trace the lineage of a single cell, we wish to determine the distribution of the number of progenitor (P) cells present. Secondly, the physical similarity between (P) and (D) cells without any protein markers make the parameters $a$, $b$ and $c$ difficult to directly measure. Thus, we would like a method to estimate them. 

\begin{wrapfigure}{r}{0.5\textwidth}
\centering
\includegraphics[height=80mm]{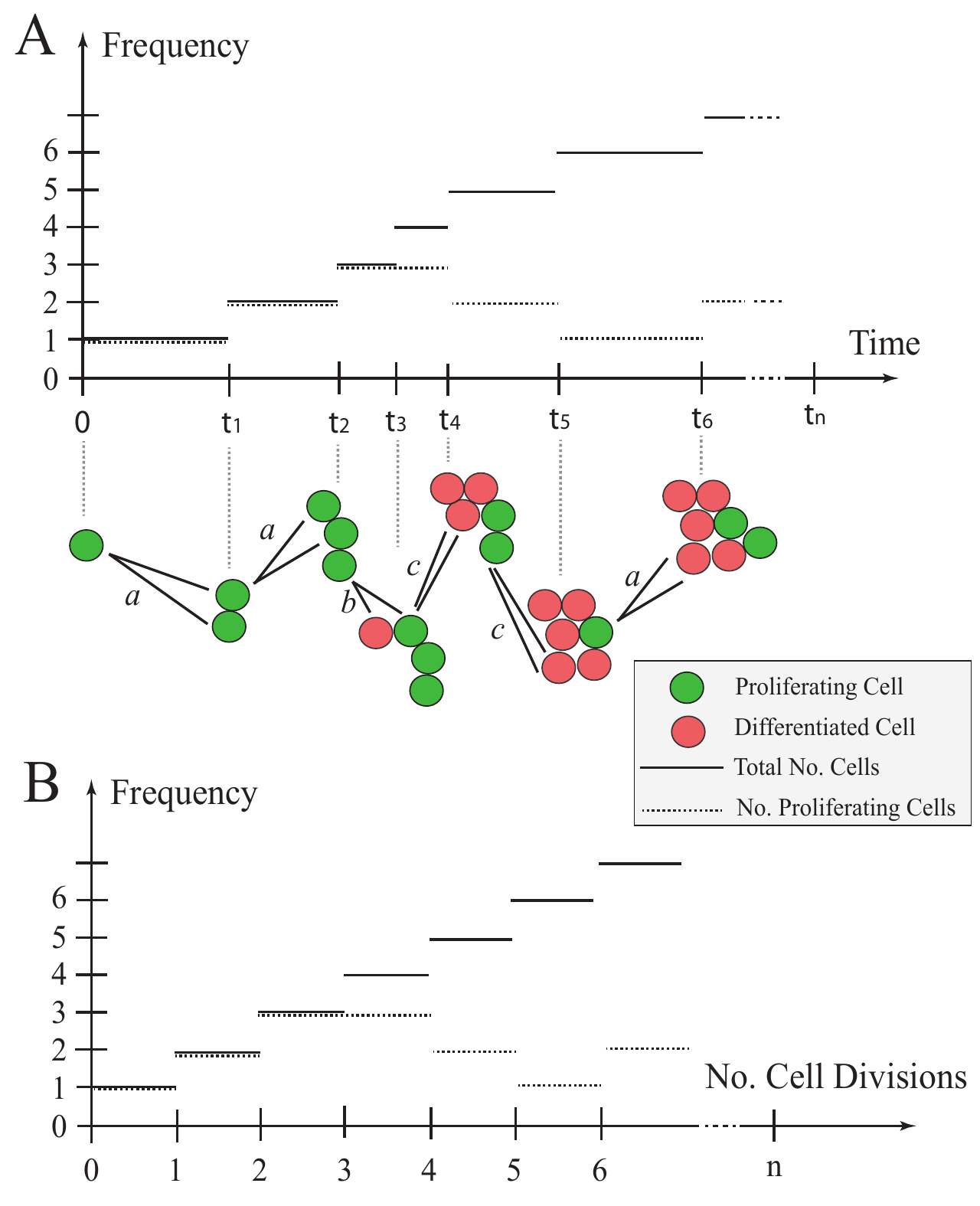}
\caption{A branching process of differentiated and proliferating cells. A single dividing cell is followed in time with the height of the solid line indicating total number of cells, and the height of the dashed line indicating number of dividing cells. In (A), plotted against time, we see the rate of cell division is dependent upon the number of proliferating cells. In (B), plotted against number of cell divisions, we see the number of proliferating cells only depends upon the nature and number of cell divisions, not their timing.}
\label{MutationDivisionB}
\end{wrapfigure}

Now, our approach is based on the size of the clone (rather than time passed). Now, with each cell division, irrespective of outcome, the colony size $n$ increases by $1$ forming a clone of $n+1$ cells. If the cell division results in two progenitor daughters (PP), the number of dividing cells $k$ increases to $k+1$. If the cell division results in a progenitor cell and a differentiated cell (PD), the number of dividing cells $k$ stays the same. The production of two differentiated daughters (DD) results in a loss of dividing cells to $k-1$. We can thus model the number of P cells as a discrete random walk that can move up, remain flat, or move down with probabilities $a$, $b$ and $c$, where we have one forward step to take at every cell division as in Figure \ref{MutationDivisionB}A,B. Note that if the colony becomes fully differentiated, $k=0$, we have no dividing cells and our process stops.

We note that the timing of these divisions does not relate to the count of proliferating cells. In Figure \ref{MutationDivisionB}A we see the time dependent process, with a division rate that will be proportional to the number of proliferating cells. In Figure \ref{MutationDivisionB}B we see the same information indexed by the number of cell divisions; the timing is not important.

Such a problem is closely related to counting Motzkin lattice paths \cite{Motzkin}. Lattice paths are paths connecting positions with integer coordinates and can take a variety of forms \cite{Mohanty}, \cite{Narayana}. In particular, Motzkin paths start from the origin $(0,0)$ on a 2-d integer lattice and allow movement with an \emph{up} $(1,1)$ step, a \emph{flat} $(1,0)$ step, or a \emph{down} $(1,-1)$ step such that we never move below the horizontal axis.  There are several path counting techniques for such conditions \cite{Motzkin}, \cite{Donaghey}, \cite{Sulanke}, which have also seen applications to paths similar to the ones we describe \cite{Lengyel} \cite{Niederhausen}. These have been studied for a range of combinatorial problems \cite{Stanley}, including some problems with weighted edges \cite{Meshkov}. 

These paths can be utilised to represent our problem. The position $(n,k)$ corresponds to the total number of cells, $n$, and the number of dividing cells, $k$, respectively. The PP, PD or DD divisions correspond to the up, flat and down steps, respectively. There are three differences to Motzkin paths to note. Firstly, we start with one (P) cell, represented by position $(1,1)$. Secondly, we stop if we touch the horizontal axis, because no dividing (P) cells remain ($k=0$). Lastly, we have probabilities $a$, $b$ and $c$ associated with each step. Now, we would like to find the probability $P_{n,k}$ of finding $k$ dividing cells in a clone of size $n$. This probability then corresponds to a weighted sum of Motzkin paths from $(1,1)$ to $(n,k)$, where Motzkin paths in this context do not touch the horizontal axis.

\begin{figure}[t!]
\centering
\includegraphics[height=100mm,width=135mm]{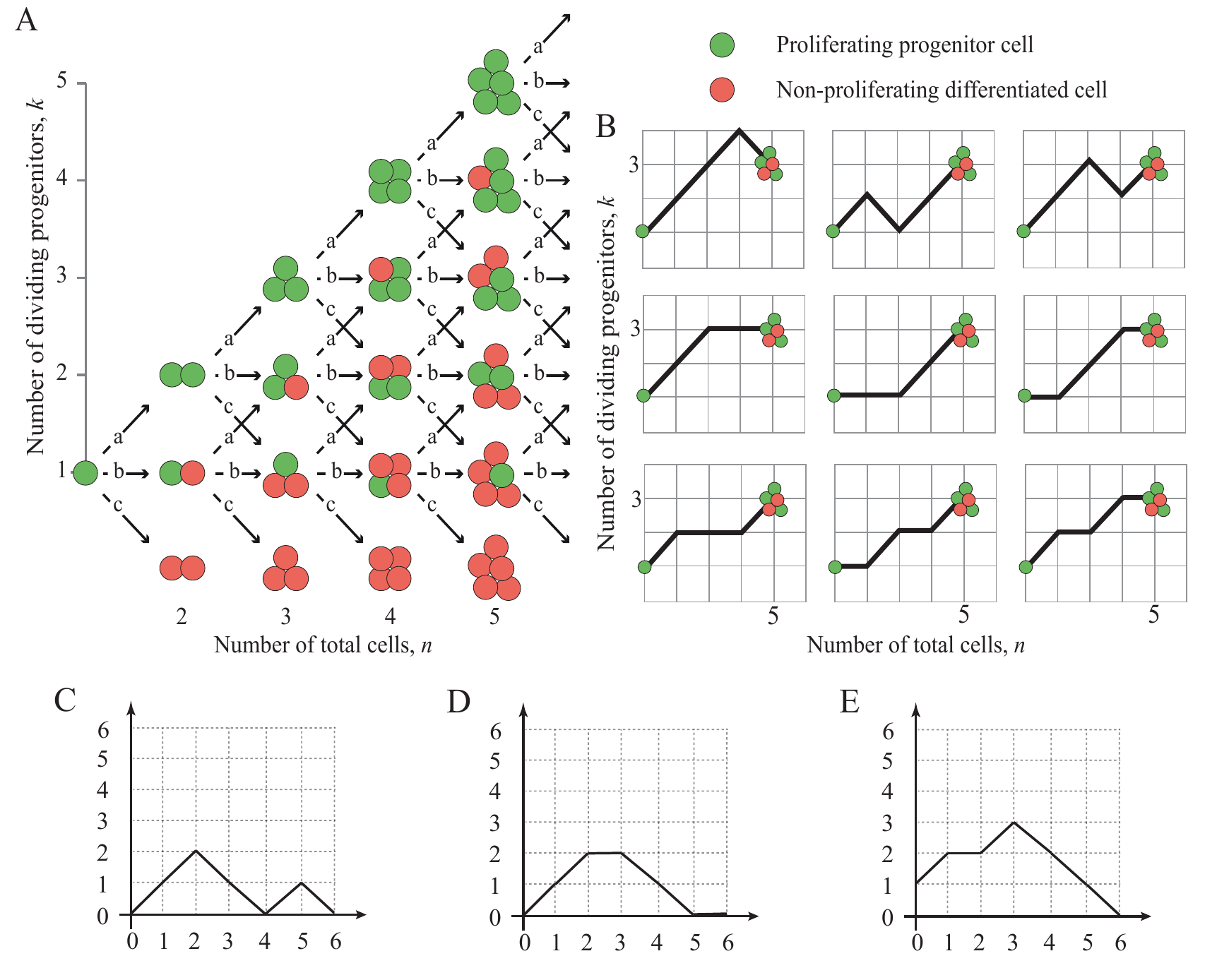}
\caption{Cell proliferation as a combinatorial branching process with predictable paths. (A) Cell division in progenitor cells is a branching process with three possible outcomes (PP, PD or DD, See Figure 1). The expansion of a single cell to form a clone of cells is thus a combinatorial process, where any outcome of total clone size $n$ and proliferating cells within it $k$ occurs along fixed paths of a Motzkin like triangle. Clones that reach the horizontal axis have only non-dividing cells, and therefore do not progress further. (B) Example showing the $9$ paths that a single proliferating cell can take to reach a clone of $n=5$ and $k=2$. The first three routes have three $a$ divisions and one $c$ division, while the remaining six routes involve two each of $a$ and $b$ divisions (cumulative probability $=3a^3c+6a^2b^2$). (C) is a Dyck path, which moves up and down and not below the horizontal axis. (D) is a Motzkin path, which also includes horizontal moves. (E) is a gamblers ruin problem, which starts from height $1$ rather than the origin, representing the formation of a fully differentiated clone from a single dividing cell.}
\label{Lattice}
\end{figure}


\subsection{Motzkin Paths Describe the Entire Distribution of Colony Sizes}

We have the following distribution for the number of progenitor (P) cells in a colony.

\begin{theorem}
If we seed a single dividing cell, then the probability of having $k(>1)$ dividing cells when the colony is of size $n$ is given by:

$P_{n,k}=\sum\limits_{i=0}^{\lfloor \frac{n-k}{2} \rfloor}{n-1 \choose k+2i-1}({k+2i-1 \choose i}-{k+2i-1 \choose i-1})a^{k+i-1}b^{n-k-2i}c^i$

\end{theorem}
\begin{proof}
We start with Dyck paths; paths from $(0,0)$ to $(0,2n)$ that do not go below the horizontal axis involving steps of type up, $(1,1)$ or down, $(1,-1)$, such as portrayed in Figure \ref{Lattice}C. The number of such paths is known to be counted by the Catalan numbers $C_{n}=\frac{1}{n+1}{2n \choose n}$ \cite{Bailey}. A Dyck triangle is the collection of paths from $(0,0)$ to $(n,k)$ that do not go below the horizontal axis and involve up and down steps. Note that $n$ and $k$ must have the same parity. If $D_{n,k}$ count these paths then conditioning over one step we find $D_{n,k}=D_{n-1,k-1}+D_{n-1,k+1}$. It is straightforward to show by substitution that $D_{n,k}=\frac{k+1}{n+1}{n+1 \choose \frac{1}{2}(n-k)}$ satisfies this recurrence, along with boundary condition $D_{2n,0}=C_{n}$. This formula differs to other counts involving Dyck triangles because this lattice formulation of the triangle is rotated through $\frac{\pi}{4}$ to the usual presentation \cite{Shapiro}.

We now turn to Motzkin paths, which are the same as Dyck paths except we now allow an additional horizontal step $(1,0)$. Now any Motzkin path from $(0,0)$ to $(n,k)$ can be partitioned into a Dyck path from $(0,0)$ to $(k+2i,k)$ involving $k+i$ up steps and $i$ down steps, along with $n-k-2i$ horizontal steps, where $i \in {0,1,...,\lfloor \frac{n-k}{2} \rfloor}$. For any $i$, the probability of such a path arising is $a^{k+i}b^{n-k-2i}c^i$. Then noting that we have ${n \choose k+2i}$ permutations of the horizontal steps with the Dyck path steps, we sum across the possibilities to get the following probability.

$m_{n,k}=\sum\limits_{i=0}^{\lfloor \frac{n-k}{2} \rfloor}D_{k+2i,k}{n \choose k+2i}a^{k+i}b^{n-k-2i}c^i=\sum\limits_{i=0}^{\lfloor \frac{n-k}{2} \rfloor}{n \choose k+2i}({k+2i \choose i}-{k+2i \choose i-1})a^{k+i}b^{n-k-2i}c^i$

Finally we note that we are going from position $(1,1)$ to $(n,k)$ without touching the horizontal axis, so substituting $n \rightarrow n-1$ and $k \rightarrow k-1$ gives the required result; $P_{n,k}=m_{n-1,k-1}$.
\end{proof}

This result allows us to look at the case where all $n$ cells in the colony are fully differentiated (all are (D) cells) and there is not further potential for growth. In our Motzkin triangle analogy, this would be a Motzkin path (with an additional final down step) from $(1,1)$ to $(n,0)$, such as the path in Figure \ref{Lattice}E. All colonies that have a corresponding path touching the horizontal axis thus have no proliferating cells. We have an absorbing barrier, also known as the gambler’s ruin problem. 

\begin{corollary}
\label{Ruin}
The probability $P_{n,0}$ is given by weighted Motzkin numbers:

$P_{n,0}=\sum\limits_{i=0}^{\lfloor \frac{n-2}{2} \rfloor}{n-2 \choose 2i}({2i \choose i}-{2i \choose i-1})a^{i}b^{n-2-2i}c^{i+1}=\sum\limits_{i=0}^{\lfloor \frac{n-2}{2} \rfloor}{n-2 \choose 2i}C_{2i}a^{i}b^{n-2-2i}c^{i+1}$

\end{corollary}
\begin{proof}
For the case where there are no dividing cells remaining in the colony, the colony must transit through a penultimate stage $(n-1,1)$ with only one dividing cell remaining, and undergo an enforced final (DD) division. Multiplying the formula for $P_{n-1,1}$ by $c$ gives the required result.
\end{proof}

Both of these results have corresponding generating functions as described in the following result:

\begin{theorem}
\label{GenFunc}
The generating function $F(x,t)=\sum\limits_{n=0}^\infty\sum\limits_{k=0}^nP_{n,k}x^kt^n$ is given by:

$F(x,t)=\frac{1-bt-\sqrt{(bt-1)^2-4act^2}}{2a}+\frac{x(2tax-1+bt+\sqrt{((bt-1)^2-4act^2))}}{2a(x-tc-tbx-tax^2)}$

\end{theorem}
\begin{proof}
First we construct a weighted generating function for paths in a standard Motzkin triangle, $m(x,t)=\sum\limits_{n=0}^\infty\sum\limits_{k=0}^nm_{n,k}x^kt^n$, where $m_{n,k}$ are the Motzkin numbers weighted by the elements $a$, $b$ and $c$ associated with each path from $(0,0)$ to $(n,k)$. Now conditioning over a single step gives the following recurrence; $m_{n+1,k}=cm_{n,k+1}+bm_{n,k}+am_{n,k-1}$. Then substituting this into the generating function yields the following:

$\begin{array}{l  l  l}
m(x,t) & = & 1+\sum\limits_{n=1}^\infty\sum\limits_{k=0}^nm_{n,k}x^kt^n =  1+\sum\limits_{n'=0}^\infty\sum\limits_{k=0}^{n'+1}m_{n'+1,k}x^kt^{n'+1}\\
& = & 1 + \sum\limits_{n'=0}^\infty\sum\limits_{k=0}^{n'+1}(cm_{n',k+1}+bm_{n',k}+am_{n',k-1})x^kt^{n'+1} \\
& = & 1 + \frac{tc}{x}(m-m(0,t)) + tbm +taxm \\
\end{array}$

Which gives us:

$m(x,t) = \frac{x-tcm(0,t)}{x-tc-tbx-tax^2}$.

To find $m(0,t)$ we note that a Motzkin path from $(0,0)$ to $(n,0)$ involves either a first horizontal step and a weighted Motzkin path one step smaller, or an up step, a motzkin path, a down step, and a Motzkin path. This is summarised as $m_{n+1}=bm_n+ac\sum\limits_{k=0}^{n-1}m_km_{n-1-k}$, where $m_n$ is the weighted sum of these paths. Now substituting this recurrence into the generating function $m(0,t) = \sum\limits_{k=0}^\infty m_kt^k=1+t\sum\limits_{k=0}^\infty m_{k+1}t^k$ yields $m(0,t) = 1 + btm(0,t) + t^2acm(0,t)^2$. The solution satisfying $m(0,0)=1$ is then:

$m(0,t) = \frac{1-bt-\sqrt{(bt-1)^2-4act^2}}{2act^2}$

Substituting this above then yields the general form:

$m(x,t)= \frac{2tax-1+bt+\sqrt{(bt-1)^2-4act^2}}{2at(x-tc-tbc-tax^2)}$ 

The result is obtained by noting that the generating function for $P_{n,k}$ corresponds to paths from $(1,1)$ to $(n,k)$. Furthermore, a path from $(1,1)$ to $(n,0)$ involves a weighted Motzkin path of length $n-2$, followed by a down step, and we find that:

$F(x,t) = \sum\limits_{n=0}^\infty P_{n,0}t^n+\sum\limits_{n,k \ge 1}P_{n,k}x^kt^n = t^2c\sum\limits_{n=0}^\infty m_{n,0}t^n+xt\sum\limits_{n,k \ge 0}m_{n,k}x^kt^n=t^2cm(0,t)+xtm(x,t)$.

Substituting the weighted Motzkin generating functions results in the desired form.
\end{proof}


\subsection{Gambler's Ruin}

We are now in a position to describe the probability of ruin, or equivalently the probability of a fully differentiated clone, where we have the following result:

\begin{corollary}
\label{RuinGF}
The generating function $G(t)=\sum\limits_{n=0}^\infty P_{n,0}t^n$ is given by:

$G(t)=\frac{1-bt-\sqrt{(bt-1)^2-4act^2}}{2a}$

This results in an alternative expression for the probability $P_{n,0}$ that a clone of size $n$ is fully differentiated:

$P_{n,0} = \frac{-(\frac{1}{2})^n}{2a}\sum_{r=0}^{n}{2(n-r)\choose{n-r}}{{2r}\choose{r}}\frac{(b+2\sqrt{ac})^{n-r}(b-2\sqrt{ac})^r}{(2(n-r)-1)(2r-1)}$

Furthermore, we find that the probability $P_{0}$ that a single proliferating cell will become fully differentiated is given by:

$P_{0} = \left\{
\begin{array}{l l}
1 & a>c \\
\frac{a}{c} & a<c
\end{array}
\right.$

\end{corollary}
\begin{proof}
To get the generating function $G(t)$, we simply substitute $x=0$ into $F(x,t)$ from Theorem \ref{GenFunc}.
To get the alternative expression for the probabilities $P_{n,0}$ note that we can write $G(t)$ as:

$G(t) = \frac{1}{2a}[1-bt-(1-(b+2\sqrt{ac})t)^{\frac{1}{2}}(1-(b-2\sqrt{ac})t)^{\frac{1}{2}}]$

A double binomial expansion gives us:

$G(t) = \frac{1}{2a}[1-bt-\sum_{j=0}^{\infty}\sum_{k=0}^{\infty}{{2j}\choose{j}}{{2k}\choose{k}}\frac{(\frac{b+2\sqrt{ac}}{2})^j(\frac{b-2\sqrt{ac}}{2})^k}{(2j-1)(2k-1)}t^{j+k}]$

The constant and linear terms cancel and a reordering of the summation to collect powers of $t$ leaves us with the required expression.

Lastly we note that $G(1)=\sum\limits_{n=0}^{\infty}P_{n,0}$ and so substituting $t=1$ into the generating function gives us:

$G(1) = \frac{1}{2a}(1-b-\sqrt{(b-1)^2-4ac} = \frac{1}{2a}(a+c-|a-c|)$

where we have used $1-b = a+c$. Separately considering the cases $a>c$ and $a<c$ gives the required results.
\end{proof}


\subsection{Estimating Differentiation Probabilities}

\begin{wrapfigure}{r}{0.5\textwidth}
\centering
\includegraphics[height=90mm]{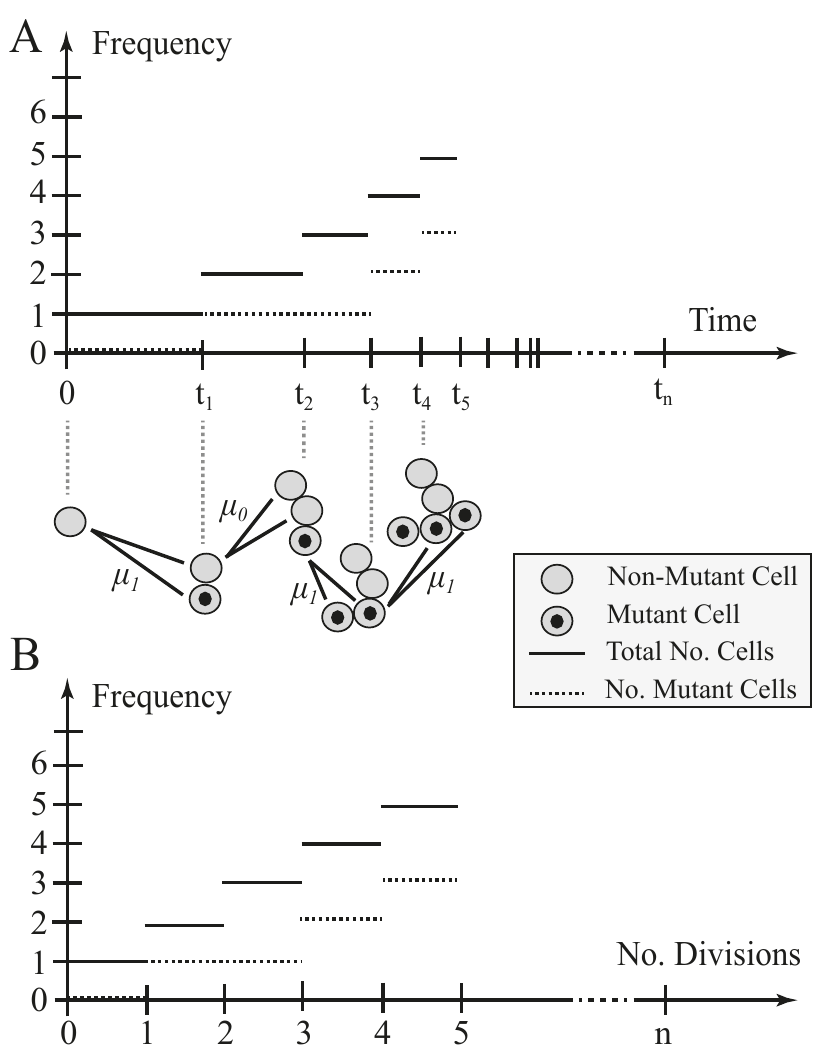}
\caption{A branching process of non-mutated and mutated cells. A single dividing cell is followed in time with the height of the solid line indicating total number of cells, and the height of the dashed line indicating number of mutant cells. In (A), plotted against time, we see the rate of cell division is proportional to the total number of cells, resulting in exponential growth. In (B), plotted against the number of divisions, we see the number of mutant cells only depends upon the number of mutant cell divisions, not their timing.}
\label{MutationDivisionC}
\end{wrapfigure}

We are now in a position to estimate the probabilities $a$, $b$ and $c$ of getting the different daughter cell combinations of (PP), (PD) or (DD), even when (P) and (D) cells are visually indistinguishable. Clone size distributions in a range of homeostatic epithelia demonstrate that dividing progenitor cells have (PP) outcomes in similar proportions to (DD) outcomes (or a=c) \cite{Jones}. Colonies arising from such populations will eventually become fully differentiated and stop growing, as represented in the bottom row of Figure \ref{Lattice}A. Therefore, at late time points of observation all colonies of cells with few cell numbers will be formed exclusively of non-dividing cells, as any colonies with dividing cells will continue to expand in cell number. Thus repeated measurements of small clone sizes, $n_c$, of fully differentiated non-dividing colonies of size $n$ can readily be counted. We can then compare these counts to the probabilities $\{c,bc,c(b^2+ac),...\}=\{P_{n,0}\}_n$ of either Corollary \ref{RuinGF} or \ref{Ruin}, and hence determine $a$, $b$ and $c$ using maximum likelihood. Small clone sizes form the bulk of clones seen in population distributions, and therefore can provide robust quantifiable results.

It is also important to highlight, that this is not affected by the presence of additional cell populations which have a branching birth process alone (putative stem cell populations). The clones formed by such populations will be much larger, continuing to expand with time, so can be readily identified and excluded. 


\subsection{Stochastic Processes Approach}

Finally, we remark that a lot of the derivations using Motzkin paths can also be replaced with approaches from stochastic processes. We highlight this with an alternative derivation of the gambler's ruin generating function of Corollary \ref{RuinGF} in the Appendix.

\section{An Exact Luria Delbr\"{u}ck Distribution}

We now investigate the mutation process of a growing clone of cells. Here, we assume no death process is involved, and initially that the mutation provides no additional survival advantage. In all that follows $k=m+n$ is the number of cells, where $m$ and $n$ count the number of mutants and non-mutants, respectively. 

Again we start with a single dividing cell. An example of this can be seen in Figure \ref{MutationDivisionC}A. The cells are dividing randomly at a rate $\beta$ according to the following Markovian branching (Yule-Furry) process. When any non-mutant cell divides we assume a mutant cell arises with probability $\mu_1$, such as the first division of Figure \ref{MutationDivisionC}A at time $t_1$. Conversely, we may obtain two non-mutants with probability $\mu_0 = 1-\mu_1$, such as in the second division portayed at time $t_2$. Finally, any dividing mutant produces two mutant daughters with probability 1, as displayed at times $t_3$ and $t_4$. We ignore any back mutation or loss of mutation.

As the colony grows, the rate of division, $\beta k$, increases in proportion to the number of cells present, $k$. If $t_k$ is the time of the $k^{th}$ division, the mean time intervals $t_{k+1}-t_k$ correspondingly decrease as we get exponential growth. Note that at time $t_k$ the colony increases in size (by one cell) to $k+1$ cells. It is this single dividing cell that has the opportunity to effect the number of mutations at this point; this is independent of the either the time $t_k$ at which this takes place, or the time $t_{k+1}-t_k$ between divisions.  

In Figure \ref{MutationDivisionC}B we see the mutation process as a discrete process on the number of divisions that have taken place. We assume for the moment that mutant and non-mutant cells divide at the same rate in a Markovian manner. All cells are thus equally likely to divide at any point in time. If we have $n$ non-mutant cells and $m$ mutant cells, we then find that a mutant will divide with probability $\frac{m}{m+n}$ resulting in $m+1$ mutants and $m+n+1$ cells. Conversely, a non-mutant divides with probability $\frac{n}{m+n}$ resulting in $m+n+1$ cells. This non-mutant will mutate with probability $\mu_1$ resulting in $m+1$ mutants, otherwise we will still have $m$ mutants, with probability $\mu_0$. This observation leads to the following correspondence.

\begin{theorem}
If $p_m^{(k)}$ denotes the probability of having $m$ mutant cells present when the population is of size $k$, we have the following recurrence, which is initialized with $p_0^{(1)}=1$.

$p_m^{(k)}=(\frac{m-1}{k-1}+\frac{k-m}{k-1}\mu_1)p_{m-1}^{(k-1)}+(\frac{k-1-m}{k-1}\mu_0)p_m^{(k-1)}$
\label{ProbRec}
\end{theorem} 
\begin{proof}
This result is a statement of conditional probability. There are two ways we can obtain $m$ mutants amongst $k$ cells, depending upon the mutation status of the $k-1$ cells prior to the previous cell division. If we have $m-1$ mutant cells out of $k-1$ cells in total, then to obtain $m$ mutants in $k$ cells we either select a mutant to divide with probability $\frac{m-1}{k-1}$, or pick a non-mutant to divide and generate a new mutation with probability $\frac{k-m}{k-1}\mu_1$. Alternatively, if we already have $m$ mutants from $k-1$ cells, then we require that the next division be a non-mutant cell that doesn't mutate on division, with probability $\frac{k-1-m}{k-1}\mu_0$.
\end{proof}

Note that we have reduced the mutation process to a discrete heterogeneous Markovian random walk starting from $(1,0)$ where we have either a horizontal step $(1,0)$ with probability $\frac{k-m}{k}\mu_0$, or the step $(1,1)$ with probability $\frac{m}{k}+\frac{k-m}{k}\mu_1$. We have the following general form for the $k$ division distribution of mutants, $p_m^{(k)}$.

\begin{theorem}\hspace*{\fill}
\\ 

$p_m^{(k)} = \mu_0^{k-1}\sum\limits_{\{1 \le i_1 < i_2 < \hdots <i_m \le k-1\}}\prod\limits_{j=1}^m\frac{(j-1)+i_j\frac{\mu_1}{\mu_0}}{k-j}$
\label{MainResult}
\end{theorem}

\begin{proof}
We start with a single cell, so the $(i-1)^{th}$ division results in $i$ cells. The required number of cells is $k$ so we have $k-1$ cell divisions in total to consider. We require $m$ mutants, so we need $m$ divisions that either generate a new mutant when a non-mutant divides, or involve a dividing mutant parent. We let the $i_j^{th}$ division be the one that increases the mutation count from $j-1$ to $j$. Then prior to this event we have $j-1$ mutant cells and $i_j$ cells. The probability that the next dividing cell is a mutant, or a non-mutant developing a mutation, is then $\frac{(j-1)+(i_j-(j-1))\mu_1}{i_j}$. For the remaining cell divisions we require a non-mutant to divide and not mutate. This will keep the mutant count constant. If this is the $i^{th}$ division, where $m$ mutants are present, this occurs with probability $\frac{i-m}{i}\mu_0$. For the first $i_1-1$ divisions there are no mutants present and this will simply be $\mu_0$. Then we find that:

$\begin{array}{l l l l}
p_m^{(k)}  &=\sum\limits_{\{1 \le i_1 < i_2 < \hdots <i_m \le k-1\}} & \mu_0^{i_1-1} & \cdot \mu_1 \cdot \frac{i_1}{i_1+1}\mu_0 \cdot \frac{i_1+1}{i_1+2}\mu_0 \cdot \hdots \cdot \frac{i_2-2}{i_2-1}\mu_0 \cdot\hdots\\
& & \hdots & \frac{1+(i_2-1)\mu_1}{i_2} \cdot \frac{i_2-1}{i_2+1}\mu_0 \cdot \hdots \cdot \frac{i_3-3}{i_3-1}\mu_0\cdot \hdots\\
& & \vdots & \\
& & \hdots & \frac{(m-1)+(i_{m-1}-(m-1))\mu_1}{i_{m-1}} \cdot \frac{i_{m-1}-m}{i_{m-1}+1}\mu_0 \cdot \hdots \cdot \frac{i_m-1-m}{i_m-1}\mu_0\\
\end{array}$

Most term then cancel to leave us with:

$p_m^{(k)} =\mu_0^{k-1}\sum\limits_{\{1 \le i_1 < i_2 < \hdots <i_m \le k-1\}}\prod\limits_{j=1}^m\frac{\mu_0^{-1}((j-1)+(i_j-(j-1))\mu_1)}{k-j}$

Noting that $\mu_0=1-\mu_1$ then gives the desired result.
\end{proof}

\begin{figure}[t!]
\centering
\includegraphics[height=42mm,width=150mm]{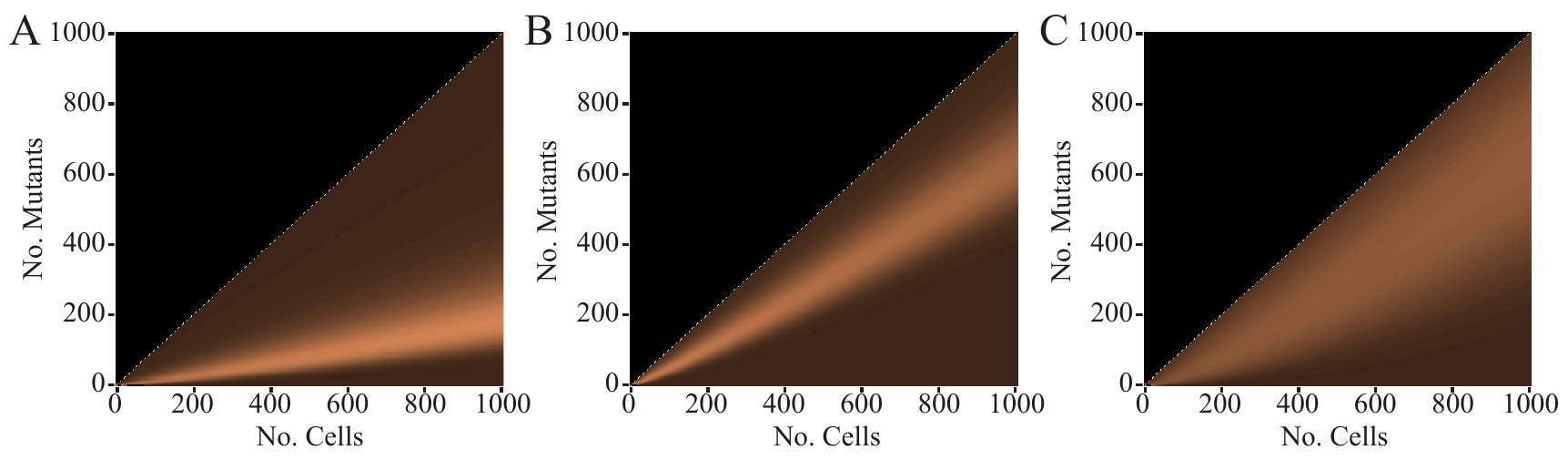}
\caption{The distributions for the number of mutants for a range of colony sizes up to 1000 cells. (A) is for $\mu_1=0.05, \rho=1$, (B) is for $\mu_1=0.2, \rho=1$, (C) is for $\mu_1=0.05, \rho=2$, where $\mu$ is the mutation rate and $\rho$ is the relative mutant fitness ($\frac{\mu_1}{\mu_0}$).}
\label{MutantDist}
\end{figure}

An example of the resulting distributions can be seen in Figure \ref{MutantDist}A,B. Although one can attempt to expand the summation and further reduce this formula, it quickly results in complicated expressions involving Faulhaber's formula for summing integer powers which do not seem to readily simplify. We have the following result concerning the moments.

\begin{theorem}
If $E_r^{(k)}$ represent the $r^{th}$ moment when $k$ cells are present, we have the following recursions for the first two moments, initialised with $E_1^{(1)}=0$ and $E_2^{(1)}=0$.

$E_1^{(k)}=\mu_1 + E_1^{(k-1)}(1+\frac{\mu_0}{k-1})$

$E_2^{(k)}=\mu_1 + E_1^{(k-1)}(2-\mu_0\frac{2k-3}{k-1})+E_2^{(k-1)}(1+\frac{2\mu_0}{k-1})$

The mean value can be written as follows.

$E_1^{(k)}=\mu_1 \sum_{u=1}^{k-1}\prod_{r=k-u+1}^{k-1}(1+\frac{\mu_0}{r})$
\end{theorem}
\begin{proof}
The recurrences for the $r^{th}$ moment are obtained by multiplying the recurrence of Theorem \ref{ProbRec} by $m^r$ and summing over $m$. The resulting equations reduce to the stated expressions after some standard algebraic manipulation. To establish the formula for the mean, we use induction on the stated result with the first recursion. For the initial values, note that the process starts with one non-mutant cell, so the mean value and second moment must both be zero.
\end{proof}


\section{Incorporating Selection}

For certain mutations, there may be a subsequent growth advantage. This has been observed with p53 mutations in epidermal tissue, for example \cite{Jones}. Our assumption that all cells are equally likely to divide is no longer valid, with mutants dividing at a different rate to non-mutants. However, we find that the mutation process is only dependent upon the ratio of these rates, and we can condition on the number of cells and apply a similar technique to the previous section to obtain the following.

\begin{theorem} Let the division rate for non-mutants and mutants be $\beta_n$ and $\beta_m$, respectively, with ratio $\rho = \frac{\beta_m}{\beta_n}$. If $p_m^{(k)}$ represents the probability of having $m$ mutant cells when there are $k=m+n$ cells present, then we have the following recurrence, initialized with $p_0^{(1)}=1$.

$p_m^{(k)}=(\frac{\rho(m-1)}{\rho(m-1)+n}+\frac{n}{n+\rho(m-1)}\mu_1)p_{m-1}^{(k-1)}+(\frac{n-1}{n-1+\rho m}\mu_0)p_m^{(k-1)}$

\end{theorem} 
\begin{proof} 
If we suppose that the mutant cells are dividing at a rate $\beta_m$ and the non-mutant cells are dividing at a rate $\beta_n$. We further suppose we have $m$ and $n$ of these cells, respectively. Then if $T_m$ is the time until the next mutant cell divides, this has exponential distribution with mean $\frac{1}{\beta_m m}$. The time $T_n$ until the next normal cell divides is similarly exponential with mean time $\frac{1}{\beta_n n}$. Then if we know we have a cell division at some point in time, we would like to know which type of cell will divide. specifically we require:

$Pr(T_m>T_n)=\int_0^{\infty} \int_0^{t_m}\frac{1}{\beta_m m}e^{\beta_m m t_m} \frac{1}{\beta_m m}e^{\beta_m m t_m} dt_n dt_m=\frac{m\beta_m}{m\beta_m + n\beta_n}=\frac{\rho m}{\rho m + n}$

Thus we just have to weight the mutant count by the relative increase in division rate. In particular, if we have $m$ mutant cells and $n-1$ non-mutant cells, the probability that we have $m$ mutants and $n$ non-mutants after the next cell division requires a non-mutant to divide without a new mutation forming. This occurs with probability $\frac{n-1}{n-1+\rho m}\mu_0$. Similarly, if we have $m-1$ mutant cells and $n$ non-mutant cells, the probability that we have $m$ mutants and $n$ non-mutants after the next cell division requires a mutant to divide, or a non-mutant to divide with a new mutation forming. This occurs with probability $\frac{\rho(m-1)}{\rho(m-1)+n}+\frac{n}{n+\rho(m-1)}\mu_1$. The recurrence is a statement of conditional probability connecting these two observations.
\end{proof}

The recurrence can be used to derive formulae for $p_m^{(k)}$ and the moments. The method is identical to that of Theorem \ref{MainResult} and the details are left to the reader. An example of the distribution can be seen in Figure \ref{MutantDist}C, where we have mutation rate $\mu_1=0.05$ and relative fitness $\rho=2$. This gave a comparable distribution to Figure \ref{MutantDist}B, where the mutation rate is $\mu_1=0.20$ with neutral relative fitness $\rho=1$, although the variance is notably higher in Figure \ref{MutantDist}C.


\section{Distributions of Sub-Clones in Mutated Colonies} 

In the questions considered so far, we just have the binary status of mutated or non-mutated. This is generally the status of a gene, or a portion of a chromosome that may be of interest, but could also be the status of a single nucleotide of DNA, which number in the billions. DNA sequencing techniques now mean that individual mutations can now be distinguished by their position in the genome. For example in Figure \ref{Depth}A we see that five of six cells are mutant, arising from four mutations produced during three cell divisions ($\dag$), that combine into into four distinct clones. In Figure \ref{Depth}B we have the distribution of the number of cells for each mutation. This is a symbolic representation of the mutation and sequencing depth information obtained from modern experiments and points to other avenues of investigation. Firstly we would like to know the number of cells containing a randomly selected mutation. Secondly, we would like to know the number of clones. Thirdly, we would like to know the number of cells. Finally, we would like to know the number of distinct mutations in a randomly selected clone. We have the following results.

\begin{figure}[t!]
\centering
\includegraphics[height=45mm,width=140mm]{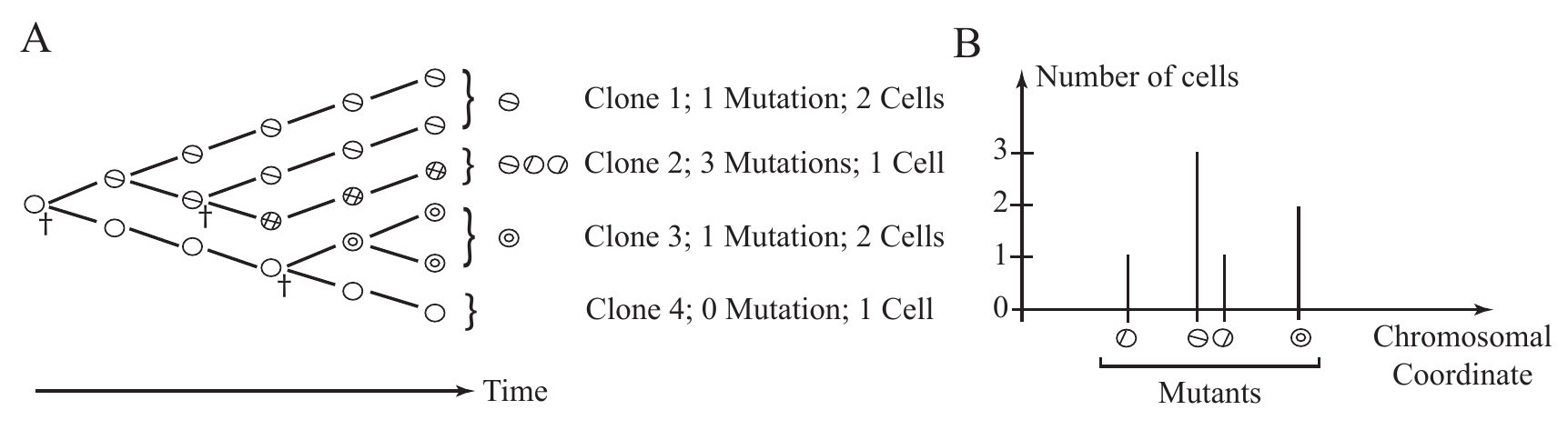}
\caption{(A) A representation of clonal mutant growth; six cells result from five cell divisions, three of which produce four mutations($\dag$), which cluster into four clones.  (B) represents the cellular count for each mutation against illustrative chromosomal co-ordinates.}
\label{Depth}
\end{figure}


\subsection{The Number of Cells Containing a Specific Mutation}

\begin{theorem}
If $p_r^{(k)}$ is the probability that a randomly selected mutation exists in $r$ cells in a colony of $k$ cells, we have:

$p_r^{(k)} = \sum\limits_{j=1}^{k-r+1}\frac{j-1}{(k-1)^2}\prod\limits_{m=1}^{r-1}(1-\frac{j-2}{k-m-1})$

\end{theorem}
This differs slightly to the original problem considered by Luria and Delbr\"{u}ck in that instead of asking how many cells contain a mutation in a specific gene (or region), which may involve many different mutation events, we randomly sample a mutation from all mutations found in that region, and count the corresponding number of cells containing that mutation. We assume each mutation arises only once, which may not be true for large colonies or small genomes.

\begin{proof}
Now there are $k-1$ divisions that take place to give a sample size of $k$. Now if we randomly select a mutation it can arise during any of these divisions with equal probability. We let $q_r^{(j,k)}$ denote the probability that if a mutation forms when there are $j$ cells, it is present in $r$ cells when the cell population is $k \ge j$. Then if $p_r^{(k)}$ is the probability a randomly selected mutation is in $r$ cells when the population is of size $k$, we have:

$p_r^{(k)}=\frac{1}{k-1}\sum\limits_{j=1}^{k-r+1}q_r^{(j,k)}$

If the mutation arises when the population has size $j$, then this mutation may be present in any of $1$ to $k-j+1$ cells, depending on whether the cells containing the mutation divide. Thus $j \le k-r+1$. Furthermore, following a population size of $k-1$, we either had $r-1$ copies of the mutation and the next cell division duplicates a copy, or we have $r$ mutant cells, and the dividing cell does not contain the mutation of interest. This gives us the recurrence:

$q_r^{(j,k)}=q_r^{(j,k-1)}(1-\frac{r}{k-1})+q_{r-1}^{(j,k-1)}(\frac{r-1}{k-1})$

Now if we start with the initial value $p_1^{(j,j)}=1$, so that initially one of $j$ cells carries the mutation, then we can show by substitution that this recurrence and initial condition is satisfied by the following expression:

$q_r^{(j,k)}=(j-1)\frac{(k-j)_{r-1}}{(k-1)_r}$

where $(a)_b=a(a-1)\hdots(a-(b-1))$ is the Pochhammer symbol. Substituting into the expression above then gives:

$p_r^{(k)} = \sum\limits_{j=1}^{k-r+1}\frac{j-1}{k-1}\frac{(k-j)_{r-1}}{(k-1)_r}$

This is equivalent to the expression in the theorem.
\end{proof}


\subsection{Distribution of the Number of Clones}

The second problem requires the distribution of the number of clones. Every time a new mutation occurs, it will occur in a single cell that belongs to some clone already present. That cell will divide into two daughters, one of which will contain the new mutation. That cell will have a new combination of mutations and a new clone is born. We thus trivially observe that the number of clones is always one more than the number of cells divisions that produce new mutations. Now, mutations can arise during a cell division. For a colony of size $k$ we have $k-1$ independent cell divisions in total, each of which may generate new mutations with probability $\mu_1$. We thus find that:

\begin{theorem}
If $C$ represents the number of colonies, we find that for a colony of size $k$, $C-1$ has Binomial distribution $Bin(k-1,\mu_1)$.
\end{theorem}
\QED


\subsection{Size Distribution of Mutant Clones}

The third question concerns the size of the clones. For example, in Figure \ref{Depth}A we note that clone $2$ was formed in the $3^{rd}$ cell division, and contains a single cell. The associated distribution for the size of a random clone is described in the following result.

\begin{theorem}
If $p_n^{(k)}$ represents the probability a randomly selected clone from a population of size $k$ contains $n$ cells, and $p_n^{(i,k)}$ is the corresponding probability for a clone formed in the $i^{th}$ cell division, then $p_n^{(k)}=\frac{1}{k-1}\sum\limits_{i=1}^{k-n}p_n^{(i,k)}$ where:

$p_n^{(i,k)} = \sum\limits_{\{i < i_1 < i_2 < \hdots <i_{n-1} \le k-1\}}\prod\limits_{j=1}^n\prod\limits_{m=1}^{i_j-i_{j-1}}(1-\frac{j}{i_{j-1}+m-1}\mu_0)\cdot\prod\limits_{j=1}^{n-1}\frac{j}{i_j}\mu_0$
where $i_0=i$.

Furthermore, $p_n^{(i,k)}$ satisfies the following recurrence:

$p_n^{(i,k)} = p_n^{(i,k-1)}\cdot(1-\frac{n}{k-1}\mu_0)+p_{n-1}^{(i,k-1)}\cdot\frac{n-1}{k-1}\mu_0$

with boundary values

$p_1^{(i,k)} = \prod\limits_{j=i+1}^{k-1}(1-\frac{1}{j}\mu_0)$
\end{theorem}
\begin{proof}
A new clone arises whenever a mutation occurs. For a colony of size $k$, a randomly selected mutation arises with equal probability $\frac{1}{k-1}$ at any of the $k-1$ divisions that have taken place. Let us suppose the colony appears at division $i$. The clone contains a single cell at this moment in time and there are $k-1-i$ divisions remaining to take place. If any of these divisions occurs in a cell not in the clone, the clone will not change in size. If the division occurs in a clone cell with mutation, the clone also does not change size because one of the two daughter cells starts a new distinct clone. However, if the division occurs in a clone cell without mutation, the clone increases in size by $1$. If the clone is of size $x$ and we have $y$ cells, this occurs with probability $\frac{x}{y}\mu_0$. Then if $p_n^{(i)}$ represents the probability that a clone starting at division $i$ ends up of size $n$ we require $n-1$ of the remaining $k-1-i$ divisions to be divisions within the clone without mutation. Then if $i_1,i_2,...,i_{n-1}$ denote the corresponding divisions, analogously to the derivation of Theorem \ref{MainResult}, we require the sum:

$\begin{array}{l l l}
p_n^{(i,k)}  &=\sum\limits_{\{i < i_1 < i_2 < \hdots <i_{n-1} \le k-1\}} & (1-\frac{1}{i+1}\mu_0)\cdot(1-\frac{1}{i+2}\mu_0)\cdot \hdots\cdot (1-\frac{1}{i_1-1}\mu_0) \cdot \frac{1}{i_1}\mu_0 \cdot\hdots\\
&  \hdots & (1-\frac{2}{i_1+1}\mu_0)\cdot(1-\frac{2}{i_1+2}\mu_0)\cdot \hdots\cdot (1-\frac{2}{i_2-1}\mu_0) \cdot \frac{2}{i_2}\mu_0\cdot \hdots\\
&  \vdots & \\
&  \hdots & (1-\frac{n-1}{i_{n-2}+1}\mu_0)\cdot \hdots\cdot (1-\frac{n-1}{i_{n-1}-1}\mu_0) \cdot \frac{n-1}{i_{n-1}}\mu_0\cdot \hdots\\
&  \hdots & (1-\frac{n}{i_{n-1}+1}\mu_0)\cdot \hdots\cdot (1-\frac{n}{k-1}\mu_0)\\
\end{array}$

Next note that a clone is equally likely to occur at any of $k-1$ cell divisions, so $p_n^{(k)}=\frac{1}{k-1}\sum\limits_{i=1}^{k-1}p_n^{(i,k)}$. Finally we note that we must have $i \le k-n$ to provide enough cell divisions to reach size $n$, otherwise $p_n^{(i,k)}$ is zero and so we obtain the required form.

For the recurrence we use a telescoping technique. First we note that we can split the sum up to give the following (where $i_0=i$, $i_n=k$ and $n \ge 2$),

$\begin{array}{l l}
p_n^{(i,k)} & = \sum\limits_{i_{n-1}=i+n-1}^{k-1}\sum\limits_{\{i < i_1 < i_2 < \hdots <i_{n-2} \le i_{n-1}-1\}} \prod\limits_{j=1}^n\prod\limits_{m=1}^{i_j-i_{j-1}}(1-\frac{j}{i_{j-1}+m-1}\mu_0)\cdot\prod\limits_{j=1}^{n-1}\frac{j}{i_j}\mu_0\\
& = \sum\limits_{i_{n-1}=i+n-1}^{k-1}\frac{n-1}{i_{n-1}}\mu_0\cdot\prod\limits_{j=i_{n-1}+1}^{k-1}(1-\frac{n}{j}\mu_0)\cdot p_{n-1}^{(i,i_{n-1})}
\end{array}$
 
From this we subtract the corresponding equation for $p_n^{(i,k-1)}\cdot(1-\frac{n}{k-1}\mu_0)$ which results in the term $p_{n-1}^{(i,k-1)}\cdot\frac{n-1}{k-1}\mu_0$. We then have the stated recursion.

For the initial value note that for a clone formed in the $i^{th}$ cell division to remain one cell in size, we must ensure that for each subsequent cell division, the single cell either does not divide (the clone remains the same), or divides with mutation (so one of the daughter cells forms a new clone). This is $1-\frac{1}{j}\mu_0$ for the $j^{th}$ division, which results in the initial condition specified in the theorem.
\end{proof}


\subsection{Number of Mutations in a Random Clone}

Finally, we need the number of mutations in a randomly selected clone. For example, note that clone $2$ from Figure \ref{Depth}A is composed of three mutations, two of which formed during the $3^{rd}$ cell division. In general we have the following result.

\begin{theorem}
Let $X_i$ be the Bernoulli variable with success probability $\frac{2}{i}$ for $i=2,3,...,k$. A clone arises at cell division $i$ with probability $\frac{1}{k-1}$, where $k$ is the total population size. The number of mutations accumulated by a clone formed in cell division $i$ is $Poisson(\lambda\sum_{j=1}^{i-1}X_i)$, where $e^{-\lambda}=\mu_0$.
\end{theorem}
\begin{proof}
New mutations occur during any cell division with a probability $\mu_1=1-\mu_0$. Now if we assume that different mutations arise independently, we can assume they are Poisson distributed per cell division with some parameter $\lambda$ so that $\mu_0=e^{-\lambda}$. Now if a clone occurs at division $i$, any subsequent mutations form new clones and do not belong to this clone. However, any earlier mutations may have been incorporated into its lineage. The first cell division occurs in this lineage with probability $1$, the second division with probability $\frac{2}{3}$, the $r^{th}$ with probability $\frac{2}{r}$. The total number of mutations in the lineage is then a sum of identical Poisson variables over cell divisions in this lineage.
\end{proof}


\section{Conclusions}
We have shown that the number of mutated or proliferating cells in a clone has a natural dependency upon the total clone size, rather than time taken for a single cell to grow into the observed clone, resulting in combinatorial and generating function approaches to analyze their distributions.

The approaches above are applicable to any Markovian branching process where daughter branches in the process have distinct characteristics. In the examples we have discussed, the branches represented the differentiation status, or the mutation status of daughter nodes. However, the branches can be more generally coloured as we like and the distributions of the number of descending nodes examined with these techniques. 

These approaches are exact but can be difficult to handle for large samples sizes and some asymptotics would be useful. Furthermore, the results all assume that the processes of cell division are Markovian and so the cell cycle exponentially distributed. This is unlikely to be accurate, with cell cycle generally being better approximated by gamma distributions. This may have significant effect on some results and warrants further exploration.

\section*{Acknowledgments}

We acknowledge support from the Cambridge Cancer Centre Research Fellowship (AR), and the Medical Research Council (PHJ).

\section*{Appendix}

Alternative proof of Corollary \ref{RuinGF} using stochastic processes methods.

\begin{proof} 

Consider a random walk which moves up or down by one unit at each step, starting from height $1$. We are interested in the number of steps taken until we first reach height $0$.

We let $u_n$ denote the probability of being at height $0$ after $n$ steps, where the walk is initially unrestricted and may move below or above height $0$. This requires $x$ up steps and $x+1$ down steps for some $x \le \frac{n-1}{2}$ and so we obtain the multinomial sum for $n \ge 1$:

$u_n = \sum\limits_{x=0}^{\lfloor \frac{n-1}{2}\rfloor}\frac{n!}{x!(x+1)!(n-2x-1)!}a^xb^{n-2x-1}c^{x+1}$

This can be used to construct an associated generating function:

$\begin{array}{l l l}
U(t) & = & \sum\limits_{n=0}^{\infty}u_nt^n = \sum\limits_{n=1}^{\infty}\sum\limits_{x=0}^{\lfloor \frac{n-1}{2} \rfloor} \frac{n!}{x!(x+1)!(n-2x-1)!}a^xb^{n-2x-1}c^{x+1}t^n \\
& = & \frac{c}{b}\sum\limits_{x=0}^\infty\sum\limits_{n=2x+1}^\infty\frac{n!}{x!(x+1)!(n-2x-1)!}(bt)^n(\frac{ac}{b^2})^x = \frac{c}{b}\sum\limits_{x=0}^\infty\sum\limits_{m=0}^\infty\frac{(m+2x+1)!}{x!(x+1)!(m)!}(bt)^{m+2x+1}(\frac{ac}{b^2})^x \\
& = & ct\sum\limits_{x=0}^\infty{{2x+1} \choose x}(act^2)^x\sum\limits_{m=0}^\infty{{m+2x+1} \choose m}(bt)^m = ct\sum\limits_{x=0}^\infty{{2x+1} \choose x}(act^2)^x\frac{1}{(1-bt)^{2x+2}} \\
& = & \frac{ct}{(1-bt)^2}\sum\limits_{x=0}^\infty{{2x+1} \choose x}(\frac{act^2}{(1-bt)^2})^x = \frac{ct}{(1-bt)^2}\frac{(1-bt)^2}{2act^2}[\frac{1}{(1-\frac{4act^2}{(1-bt)^2})^{\frac{1}{2}}}-1] \\
& = &\frac{1}{2at}[(1-\frac{4act^2}{(1-bt)^2})^{-\frac{1}{2}}-1]
\end{array}$

Here we have used the identity $2z\sum\limits_{x=0}^\infty{2x+1 \choose x}z^x=(1-4z)^{-\frac{1}{2}}-1$ on the penultimate line.

Similarly, we let $v_n$ denote the probabililty of being at height $0$ after $n$ steps, this time starting from height $0$. Again, we do not prohibit negative heights. This requires $x$ up steps and $x$ down steps for some $x \le \frac{n}{2}$ and so we obtain the multinomial sum for $n \ge 1$:

$v_n = \sum\limits_{x=0}^{\lfloor \frac{n}{2}\rfloor}\frac{n!}{x!x!(n-2x)!}a^xb^{n-2x}c^x$

This also has an associated generating function:

$\begin{array}{l l l}
V(t) & = & \sum\limits_{n=0}^{\infty}v_nt^n = \sum\limits_{n=0}^{\infty}\sum\limits_{x=0}^{\lfloor \frac{n}{2} \rfloor} \frac{n!}{(x!)^2(n-2x)!}a^xb^{n-2x}c^xt^n \\
& = & \sum\limits_{x=0}^{\infty}\sum\limits_{n=2x}^\infty \frac{n!}{(x!)^2(n-2x)!}(\frac{ac}{b^2})^x(bt)^n = \sum\limits_{x=0}^{\infty}\sum\limits_{m=0}^\infty \frac{(m+2x)!}{(x!)^2m!}(\frac{ac}{b^2})^x(bt)^{m+2x}\\
& = & \sum\limits_{x=0}^{\infty}\sum\limits_{m=0}^\infty {{m+2x} \choose m}{2x \choose x}(act^2)^x(bt)^m = \sum\limits_{x=0}^{\infty}{2x \choose x}(act^2)^x\sum\limits_{m=0}^\infty {{m+2x} \choose m}(bt)^m \\
& = & \sum\limits_{x=0}^{\infty}{2x \choose x}(act^2)^x\frac{1}{(1-bt)^{2x+1}} = \frac{1}{1-bt}\sum\limits_{x=0}^{\infty}{2x \choose x}(\frac{act^2}{(1-bt)^2})^x \\
& = & \frac{1}{1-bt}(1-\frac{4act^2}{(1-bt)^2})^{-\frac{1}{2}} \\
\end{array}$

We are interested in the first visit to height $0$ starting from height $1$. Now, if we know we are at height $0$ after $n$ steps, then there must be a first visit to height zero after $r$ steps for some $r$ with $1 \le r \le n$. If $f_r$ represents the probability of a first visit to $0$ after $r$ steps we then have the discrete convolution:

$u_n = \sum_{r=1}^nf_rv_{n-r}$

Multiplying by $t^n$ and summing then results in the following relation between generating functions:

$U=FV$

where $F=\sum_{r=0}^{\infty}f_rt^r$ is the generating funtion for the probabilities $f_r$ we desire. Then substituting the generating functions above yields the following:

$F(t)=\frac{1-bt}{2at}[1-(1-\frac{4act^2}{(1-bt)^2})^{\frac{1}{2}}]$

To obtain the required expression in Corollary \ref{RuinGF}, we note that the generating function $G(t)=\sum\limits_{n=0}^\infty P_{n,0}t^n$ relates to the probability of ruin $P_{n,0}$ when there are $n$ cells present. We start from $1$ cell, so this involves $n-1$ steps and we find $f_{n-1}=P_n$. In terms of generating functions, we find $G(t)=tF(t)$, which gives the desired form for $G(t)$.
\end{proof}

\footnotesize{

}


\begin{thebibliography}{9}

\bibitem{Clayton} Clayton E., Doupé D.P., Klein A.M., Winton D.J., Simons B.D., and Jones P.H. 2007 A single type of progenitor cell maintains normal epidermis. \emph{Nature} {\bf 446} 185–189. (DOI 10.1038/nature05574).
\bibitem{Doupe} Doupé D.P., Alcolea M.P., Roshan A., Zhang G., Klein A.M., Simons B.D., and Jones P.H. 2012 A single progenitor population switches behavior to maintain and repair esophageal epithelium. \emph{Science} {\bf 337}, 1091–1093. (DOI 10.1126/science.1218835).
\bibitem{Mascre} Mascré G., Dekoninck S., Drogat B., Youssef K.K., Brohée S., Sotiropoulou P.A., Simons B.D., and Blanpain C. 2012 Distinct contribution of stem and progenitor cells to epidermal maintenance. \emph{Nature} {\bf 489}, 257-62. (DOI 10.1038/nature11393).
\bibitem{Klein} Klein A.M., Nakagawa T., Ichikawa R., Yoshida S., and Simons B.D. 2010 Mouse germ line stem cells undergo rapid and stochastic turnover. \emph{Cell Stem Cell} {\bf 7}, 214–224. (DOI 10.1016/j.stem.2010.05.017).
\bibitem{Snippert} Snippert H.J., van Der Flier L.G., Sato T., van Es J.H., van Den Born M., Kroon-Veenboer C., Barker N., Klein A.M., van Rheenen J., Simons B.D., et al. 2010 Intestinal crypt homeostasis results from neutral competition between symmetrically dividing Lgr5 stem cells. \emph{Cell} {\bf 143}, 134-44. (DOI 10.1016/j.cell.2010.09.016).
\bibitem{Barrandon} Barrandon Y., and Green H. 1987 Three clonal types of keratinocyte with different capacities for multiplication. \emph{Proc. Natl. Acad. Sci. U. S. A.} {\bf 84}, 2302-6. 
\bibitem{Kretzschmar} Kretzschmar K., \& Watt F.M. 2012 Lineage tracing. \emph{Cell} {\bf 148}, 33–45. (DOI 10.1016/j.cell.2012.01.002).
\bibitem{Alcolea} Alcolea M.P., and Jones P.H. 2013 Tracking cells in their native habitat: lineage tracing in epithelial neoplasia. \emph{Nat. Rev. Cancer.} {\bf 13}, 161-71. (DOI 10.1038/nrc3460).
\bibitem{Blanpain} Blanpain C., and Simons B.D. 2013 Unravelling stem cell dynamics by lineage tracing. \emph{Nat. Rev. Mol. Cell Biol} {\bf 14}, 489–502. (DOI 10.1038/nrm3625).
\bibitem{Klein2} Klein A.M., and Simons B.D. 2011 Universal patterns of stem cell fate in cycling adult tissues. \emph{Development} {\bf 138}, 3103-11. (DOI 10.1242/dev.060103).
\bibitem{Jones} Klein A.M., Brash D.E., Jones P.H., and Simons B.D. 2010 Stochastic fate of p53-mutant epidermal progenitor cells is tilted toward proliferative by UV B during preneoplasia. \emph{Proc. Natl. Acad. Sci. U. S. A.} {\bf 107}, 270–275. (DOI 10.1073/pnas.0909738107).
\bibitem{Greco} Rompolas P., Deschene E.R., Zito G.,Gonzalez D.G., Saotome I., Haberman A.M., and Greco V. 2012 Live imaging of stem cell and progeny behaviour in physiological hair-follicle regeneraition. \emph{Nature} {\bf 487}, 496-499. (DOI 10.1038/nature11218).
\bibitem{Youssef} Youssef K.K., Van Keymeulen A., Lapouge G., Beck B., Michaux C., Achouri Y., Sotiropoulou P.A., and Blanpain C. 2010 Identification of the cell lineage at the origin of basal cell carcinoma. \emph{ Nat. Cell Biol.} {\bf 12}, 299–305. (DOI 10.1038/ncb2031).
\bibitem{Lapouge} Lapouge G., Youssef K.K., Vokaer B., Achouri Y., Michaux C., Sotiropoulou P.A., and Blanpain C. 2011 Identifying the cellular origin of squamous skin tumors. \emph{Proc. Natl. Acad. Sci. U. S. A.} {\bf 108}, 7431-7436. (DOI 10.1073/pnas.1012720108).
\bibitem{Schepers} Schepers A.G., Snippert H.J., Stange D.E., van Den Born M., van Es J.H., van De Wetering M., and Clevers H. 2012 Lineage tracing reveals Lgr5+ stem cell activity in mouse intestinal adenomas. \emph{Science} {\bf 337}, 730-735. (DOI 10.1126/science.1224676).
\bibitem{Barker} Barker N., Ridgway R.A., van Es J.H., van de Wetering M., Begthel H., van den Born M., Danenberg E., Clarke A.R., Sansom O.J., and Clevers H. 2009 Crypt stem cells as the cells-of-origin of intestinal cancer. \emph{Nature} {\bf 457}, 608-611. (DOI 10.1038/nature07602).
\bibitem{LuriaDelbruck} Luria S.E., and Delbr\"{u}ck M. 1943 Mutations of bacteria from virus sensitivity to virus resistance. \emph{Genetics} {\bf 28}, 491-511.
\bibitem{LeaCoulson} Lea D.E., and Coulson C.A. 1949 The distribution of the number of mutants in bacterial populations. \emph{J. Genetics} {\bf 49}, 264-285.
\bibitem{Bartlett} Bartlett M. 1978 \emph{An introduction to the stochastic process}, 3rd ed. Cambridge: Cambridge University Press. 
\bibitem{Bartlett2} Armitage P. 1952 The statistical theory of bacterial populations subject to mutation. {\emph J. Royal Statist. Soc. B} {\bf 14}, 1-33.
\bibitem{Haldane} Sarkar S. 1991 Haldane's solution of the Luria-Delbr\"{u}ck distribution.\emph{Genetics} {\bf 127}, 257-261.
\bibitem{Nik-Zainal}  Nik-Zainal S., Van Loo P., Wedge D.C.,  Alexandrov L.B., Greenman C.D.,  Lau K.W., Raine K., Jones D.,  Marshall J., Hinton J., et al. 2012 The life history of 21 breast cancers. \emph{Cell} {\bf 149}, 994-1007. (DOI 10.1016/j.cell.2012.04.023).
\bibitem{Kendall} Kendall D.G. 1960 Birth-and-death process and the theory of carcinogenesis. \emph{Biometrika} {\bf 47}, 13-21. (DOI 10.1093/biomet/47.1-2.13).
\bibitem{Zheng} Zheng Q. 1999 Progress of a half century in the study of the Luria-Delbr\"{u}ck distribution. \emph{Mathematical Biosciences} {\bf 162}, 1-32.
\bibitem{Stahl} St\r{a}hl P.L., Stranneheim H., Aspland A., Berglund L., Pont\'{e}n F., and Lundeberg J. 2011 \emph{J. Investigative Dermatology} {\bf 131}, 504-508. (DOI 10.1038/jid.2010.302).
\bibitem{Bont} DeBont R., and Larebeke N.V. 2004 Endogenous DNA damage in humans: a review of quantitative data. \emph{Mutagenesis} {\bf 19}, 3, 169-185. (DOI 10.1093/mutage/geh025).
\bibitem{Greenman} Greenman C.D., Cooke S.L., Marshall J., Stratton M.R., and Campbell P.J. 2013 Modelling breakage-fusion-bridge cycles as a stochastic folding process. \emph{arXiv:1211.2356 [q-bio.GN]}.
\bibitem{McClintock} McClintock, B. 1941. The stability of broken ends of chromosomes in Zea mays. \emph{Genetics} {\bf 26}: 234–282.
\bibitem{Kampen} Kampen N.G.V. 2007 \emph{Stochastic processes in physics and chemistry}. Amsterdam: Elsevier.
\bibitem{Motzkin} Motzkin T. 1948 Relations between hypersurface cross ratios, and a combinatorial formula for partitions of a polygon, for permanent preponderance, and for non-associative products \emph{Bull. Amer. Math. Soc.} {\bf 54}, 352-360.
\bibitem{Mohanty} Mohanty S.G. 1979 \emph{Lattice Paths Counting and Applications}. New York: Academic Press.
\bibitem{Narayana} Narayana T.V. 1979 \emph{Lattice Paths Combinatorics with Statistical Applications}. Toronto: University of Toronto Press.
\bibitem{Donaghey} Donahey R. and Shapiro L.W. 1977 Motzkin Numbers. \emph{J. Comb Theory Ser. A} {\bf 23}, 291-301.
\bibitem{Sulanke} Sulanke R.A. 2001 Bijective recurrences for Motzkin paths \emph{Advances in Applied Mathematics} {\bf 27}, 627–640.
\bibitem{Lengyel} Lengyel T. 2011 Gambler's ruin and winning a series by \emph{m} games. \emph{Ann Inst Stat Math} {\bf 63}, 181-195.
\bibitem{Niederhausen} Niederhausen H. 1998 Lattice paths between diagonal boundaries. \emph{Elec J. Comb} {\bf 5}, \#R30.
\bibitem{Stanley} Stanley R.P. 1999 \emph{Enumerative Combinatorics} vol. 2. Cambridge: Cambridge University Press.
\bibitem{Meshkov} Meshkov V.R., Omelchenko A.V., Petrov M.I., and Tropp E,A. 2010 Dyck and Motzkin triangles with multiplicities. \emph{Moskow Math. J.} {\bf 10}, 611-628.
\bibitem{Bailey} Bailey D.F. 1996 Counting arrangements of 1's and -1's. \emph{Math. Mag.} {\bf 69}, 128-131.
\bibitem{Shapiro} Shapiro L.W. 1976 A Catalan triangle. \emph{Discrete Mathematics} {\bf 14}, 83-90.
\bibitem{Motzkin2} Motzkin Triangle, Sequence A026300 [Internet]. Online Encycl. Integer Seq. 2010. Available from: http://oeis.org/A026300.

\end{thebibliography}
\end{document}